\documentclass[11pt]{article}

\usepackage{fullpage}

\usepackage{amsmath}
\usepackage{amsthm}
\usepackage{amssymb}
\usepackage{amsfonts}
\usepackage{subfigure}
\usepackage{color}
\usepackage{makeidx}
\usepackage{footnote}
\makesavenoteenv{tabular}

\newcommand{\ignore}[1]{}

\newtheorem{theorem}{Theorem}

\newtheorem{lemma}{Lemma}

\newtheorem{corollary}[theorem]{Corollary}

\renewcommand{\Pr}{{\bf Pr}}
\newcommand{\E}{{\bf E}}

\newcommand{\dist}{{\rm dist}}
\newcommand{\FF}{\mathbb{F}}

\newcommand{\wt}{{\rm wt}}

\newcommand{\bias}{{\rm bias}}

\begin{document}

\title{Derandomizing Chernoff Bound with Union Bound\\
with an Application to $k$-wise Independent Sets}
\author{ Nader H. Bshouty
\\ Technion, Haifa, Israel\\
bshouty@ca.technion.ac.il}


%
\maketitle

\begin{abstract}
Derandomization of Chernoff bound with union bound is already proven in many papers.
We here give another explicit version of it that obtains a construction of size
that is arbitrary close to the probabilistic nonconstructive size.

We apply this to give a new simple polynomial time constructions of
almost $k$-wise independent sets. We also give almost tight lower bounds
for the size of $k$-wise independent sets.
\end{abstract}

\section{Introduction}
Derandomization of Chernoff bound with union bound is already proven in many papers.
See for example~\cite{WX08,AG10,AHK12}.
We here give another explicit version of it that obtains a construction
of a size that is arbitrary close to the size of the
probabilistic nonconstructive size.

We then show that, for some
construction problems, one can combine this method with the method of conditional probabilities
to get a derandomization that runs in time polylogarithmic in the
sample size.

In this paper we give the following application of this result:

For $p\in \Re$, the {\it distance in the $L_p$-norm between two probability distribution}
$D$ and $Q$ over a sample space $S$ is
$$\|D-Q\|_p=\left(\sum_{s\in S}|D(s)-Q(s)|^p\right)^{\frac{1}{p}}$$
and for $p=\infty$ is $\|D-Q\|_p=\max_{s\in S}|D(s)-Q(s)|$.

Let $S=\{0,1\}^n$. A {\it uniform distribution $U$ over} $S$
is a distribution where $U(s)=1/2^n$ for all $s \in S$.
For $I=(i_1,\ldots,i_k)$ where
$1\le i_1<\cdots<i_k\le n$ the {\it distribution $D_I$ restricted to} $I$
over $\{0,1\}^k$ is $D_I(\sigma_1,\ldots,\sigma_k)=\Pr_{s\sim D}[s_{i_1}=\sigma_1\wedge
\cdots\wedge s_{i_k}=\sigma_k]$. A distribution $Q$ over $\{0,1\}^n$ is called
{\it $\epsilon$-almost $k$-wise independent in the $L_p$-norm}
if for any
$I=(i_1,\ldots,i_k)$ we have $\|Q_I-U_I\|_p\le \epsilon$.

The goal is to construct $S'\subset \{0,1\}^n$ of small size such that the
uniform distribution on $S'$ is $\epsilon$-almost $k$-wise independent in the $L_p$-norm.
We will just say that $S'$ is {\it $\epsilon$-almost $k$-wise independent set} in the $L_p$-norm.
The following table summarizes the results from the literature and our results.
The table shows the sizes without the small terms  $\log\log n$, $\log(1/\epsilon)$ and $k$.
See the exact sizes in the table in Section~\ref{ST} and the theorems in Subsection~\ref{Theorems}
and Section~\ref{LowerB}.

\begin{center}
\begin{tabular}{|c|c|c|c|c|}
Construction Time &Reference & Size for $L_\infty$ & Size$^*$ for $L_\infty$ & Size for $L_1$\\
\hline
Poly. time & \cite{AGHP90} & $\frac{\log^2n}{\epsilon^2}$
& $\frac{2^{2k}\log^2n}{\epsilon^2}$
& $\frac{2^k\log^2n}{\epsilon^2}$\\
\hline
Poly. time & AGC+Hadamard& $\frac{\log n}{\epsilon^3}$
&$\frac{2^{3k}\log n}{\epsilon^3}$
& $\frac{2^{3k/2}\log n}{\epsilon^3}$\\
\hline
Poly. time & \cite{BS13} & $\frac{\log^{5/4} n}{\epsilon^{2.5}}$
&$\frac{2^{2.5k}\log^{5/4} n}{\epsilon^{2.5}}$
&$\frac{2^{5k/4}\log^{5/4} n}{\epsilon^{2.5}}$\\
\hline
Poly. time & Ours & $\frac{\log n}{2^k\epsilon^3}$
& $\frac{2^{2k}\log n}{\epsilon^3}$
&$\frac{2^k+\log n}{\epsilon^3}$\\
\hline
$n^{O(k)}$ time& Ours & $\frac{\log n}{2^k\epsilon^2}$& $\frac{2^k\log n}{\epsilon^2}$&$\frac{2^k+\log n}{\epsilon^2}$\\
\hline \hline
Lower Bound & Ours & $\frac{\log n}{2^k\epsilon^2}$& $\frac{2^k\log n}{\epsilon^2}$& $\frac{\log n}{\epsilon^2}$\\
\hline
\end{tabular}
\end{center}
{\bf Table} 1. * {\small For $L_\infty$ (column 3) we have $\max_{I,\sigma} |\Pr[s_I=\sigma]-1/2^k|\le \epsilon$ for $\epsilon<1/2^k$.  We also have added the bounds (column 4) when $(1/2^k)(1-\epsilon)\le \max_{I,\sigma} \Pr[s_I=\sigma]\le (1/2^k)(1+\epsilon)$ for any $\epsilon<1$. For $L_1$ (column 5) we have $\sum_{\sigma\in \{0,1\}^k}|\Pr[s_I=\sigma]-1/2^k|\le \epsilon$ for any $\epsilon<1$.}
$$ $$

With the techniques used in this paper,
all the results in this table can be easily generalized to any product
distribution and any alphabet.

Our construction is very simple. We first give a
derandomization of Chernoff bound with union bound. For this we use the pessimistic method with a pessimistic estimator (potential function) that gives constructions of size that are arbitrary close to the size of the
probabilistic nonconstructive size. Those constructions are polynomial in the space size that is exponential in the dimension of the problem. We then use the conditional probability method that reduces the complexity to polylogarithmic time in the sample size. Those are used to construct a dense perfect hash family and an $\epsilon$-almost $k$-wise independent set of small dimension. We then combine both constructions to get the final construction. We also give lower bounds that are almost tight
to the non-constructive constructions. Our constructions have sizes that are within a factor of $1/\epsilon$ from the lower bounds.

Our construction can be easily generalized to any product distribution and any alphabet (not necessarily alphabet of size power of prime) and can be used for other dense and balance constructions. See some other techniques for deterministic and randomized dense, balance and non-dense constructions in~\cite{AG09,AG10,AM06,B14b,BG,IK,KM94,PR,NSS} and references within.

This paper is organized as follows. In Section~\ref{Dera} we give the
main two theorems of the derandomization and show how to use the method of conditional
probabilities to reduce its time complexity to polylogarithmic in the size of the
sample space. In Section~\ref{LCC} we give an exponential time constructions
of small size for a code that achieves the Gilbert-Varshamov bound, $\epsilon$-balance
error-correcting code and $\epsilon$-bias sample space. In Section~\ref{ST} we
give all the constructions in the above table. Then in Section~5 we give the lower bounds.

\section{The Derandomization}\label{Dera}

In this section we give the derandomization of Chernoff bound with union bound

The $q$-ary entropy function is
$$H_q(p)=p\log_q \frac{q-1}{p}+(1-p)\log_q\frac{1}{1-p}.$$
The Kullback-Leibler divergence between Bernoulli distributed random variables
with parameter $\lambda$ and $\eta$ is
$$D(\lambda||\eta)=\lambda\ln\frac{\lambda}{\eta}+(1-\lambda)\ln\left(\frac{1-\lambda}{1-\eta}\right).$$

For two integers $n$ and $m$ we denote $[n]=\{1,2,\ldots,n\}$ and $(n,m]=\{n+1,n+2,\ldots,m\}$. For a finite multiset of objects $S$ we denote by $U_S$ the uniform distribution over $S$.

We prove
\begin{theorem}\label{Theorem1} Let $S$ be a finite sample space
with a probability distribution $D$.
Let $X_1,\ldots,X_{N}$ be random variables over
$S$ that take values from $\{0,1\}$.
Let $N'\le N$ and $\{\lambda_i\}_{i\in [N]}$ be such that
$0<\lambda_i<p_{i}\le \E_{s\sim D}[X_i]$ and $1>\lambda_{j}>p_{j}\ge \E_{s\sim D}[X_j]$
for all $i\in [N']$ and $j\in (N',N]$.
Let $m$ be such that
$$P_0:=\sum_{i=1}^N e^{- D(\lambda_i||p_i)\cdot m}\le 1.$$
There exists a multiset $S'=\{s_1,s_2,\ldots,s_{m}\}\subseteq S$ such that for
all $i\in[N']$ and $j\in(N',N]$
$$\underset{s\sim U_{S'}}\E[X_i(s)]\ge \lambda_i,\ \ \ \underset{s\sim U_{S'}}\E[X_j(s)]\le \lambda_j.$$
In particular, the result follows for
$$m\ge\max_{i\in [N]}\frac{\ln N}{ D(\lambda_i||p_i)}=\max_{i\in[N]}\frac{\log_{q_i} N}{1-H_{q_i}(\lambda_i)}$$
where $q_i=1/(1-p_i)$.
\end{theorem}
\begin{proof}
We give an algorithm that constructs $S'$. Let
$$\alpha_i=\frac{\lambda_i}{(q_i-1)(1-\lambda_i)}=\frac{(1-p_i)\lambda_i}{p_i(1-\lambda_i)}
\ \ \mbox{and} \ \ \ \gamma_i=\frac{1}{\alpha_i p_i+(1-p_i)}=\frac{1-\lambda_i}{1-p_i}.$$
Suppose that the algorithm has already chosen
$s_1,\ldots,s_\ell$.
Consider the potential function
$$P_{j}:=\sum_{i=1}^{N} e^{-D(\lambda_i||p_i)m}\gamma_i^j \alpha_i^{Z_{j,i}}$$
where $Z_{j,i}=X_i(s_1)+\cdots+X_i(s_j)$ for $j=0,1,\ldots,\ell$.
Here, $Z_{0,i}=0$.
We now show how the algorithm chooses $s_{\ell+1}$. Consider the random variable
$$P'_{\ell+1}(s)=\sum_{i=1}^{N} e^{-D(\lambda_i||p_i)m}\gamma_i^{\ell+1} \alpha_i^{Z_{\ell,i}+X_i(s)}.$$
Since $\alpha_i<1$ for $i\in[N']$ and $\alpha_j>1$ for $j\in(N',N]$, we have
 \begin{eqnarray}\label{eq1}
\underset{s\sim D}\E[P'_{\ell+1}(s)]&=& \sum_{i=1}^{N} e^{-D(\lambda_i||p_i)m}\gamma_i^{\ell+1} \alpha_i^{Z_{\ell,i}}\underset{s\sim D}\E\left[\alpha_i^{X_i(s)}\right]\nonumber\\
&\le& \sum_{i=1}^{N} e^{-D(\lambda_i||p_i)m}\gamma_i^{\ell} \alpha_i^{Z_{\ell,i}}\nonumber\\
&=&P_{\ell}\nonumber.
\end{eqnarray}

Now the algorithm chooses $s_{\ell+1}\in S$ that satisfies $P'_{\ell+1}(s_{\ell+1})\le P_{\ell}.$
Let $P_{\ell+1}=P'_{\ell+1}(s_{\ell+1})$.
Then $P_0\le 1$ and $P_{\ell+1}\le P_{\ell}.$
Therefore

$$P_m=\sum_{i=1}^{N} e^{-D(\lambda_i||p_i)m}\gamma_i^m \alpha_i^{Z_{m,i}}\le 1 .$$
In particular, for all $i\in[N]$
$$e^{-D(\lambda_i||p_i)m}\gamma_i^m \alpha_i^{Z_{m,i}}\le 1.$$
Therefore, for all $i\in[N]$,
\begin{eqnarray*}
\alpha_i^{Z_{m,i}}\le (\gamma_i^{-1}\cdot e^{D(\lambda_i||p_i)})^m=   \alpha_i^{\lambda_i m}.
\end{eqnarray*}
Thus, for all $i\in [N]$ we have $$\alpha_i^{Z_{m,i}}\le \alpha_i^{\lambda_i m}.$$
Since $\alpha_i<1$ for all $i\in [N']$ and $\alpha_{j}>1$ for all $j\in (N',N]$ we have
$Z_{m,i}\ge \lambda_i m$ for all $i\in [N']$ and $Z_{m,j}\le \lambda_{j} m$ for all $j\in (N',N]$.
Since $\E_{s\sim U_{S'}}[X_k(s)]=Z_{m,k}/m$ for all $k\in [N]$, the result follows.
\end{proof}

We now give the bit-time complexity of the algorithm described in the proof of Theorem~\ref{Theorem1}
\begin{theorem} \label{Theorem2}
Let all notation and assumptions be as in Theorem~\ref{Theorem1}.
Suppose that for every $s\in S$ all the values $X_1(s),\ldots,X_N(s)$
can be computed in bit-time\footnote{Here $\tilde O(N)$ is
$O(N\cdot poly(\log T))$ where $T$ is the time complexity
of the construction.}  $\tilde O(N)$. Let $\mu_i=D(\lambda_i||p_i)$, $\mu=\min(1,\min_i\mu_i)$.
Let $\tau=\max_i\max(\alpha_i,\gamma_i)$ where $$\alpha_i=\frac{(1-p_i)\lambda_i}{p_i(1-\lambda_i)}
\ \ \mbox{and} \ \ \ \gamma_i=\frac{1-\lambda_i}{1-p_i}.$$
There is an algorithm that runs in bit-time $$T=\tilde O(|S|\cdot Nm\cdot (m\log \tau+\log (1/\mu))$$
and outputs a multiset $S'=\{s_1,\ldots,s_m,s_{m+1}\}\subseteq S$ such that for
all $i\in[N']$ and $j\in(N',N]$
$$\underset{s\sim U_{S'}}\E [X_i(s)]\ge \lambda_i,\ \ \underset{s\sim U_{S'}}\E[X_j(s)]\le \lambda_j$$
where $U_{S'}$ is the uniform distribution on $S'$.
\end{theorem}
\begin{proof} Consider the algorithm in the proof of Theorem~\ref{Theorem1}.
Notice that here the size of $S'$ is $m+1$ and not $m$ as in Theorem~\ref{Theorem1}.
So the potential function used in the algorithm is
$$P_{j}:=\sum_{i=1}^{N} e^{-\mu_i (m+1)}\gamma_i^j \alpha_i^{Z_{j,i}}$$
but with the same assumption
$$\sum_{i=1}^{N} e^{-\mu_i m}\le 1$$ as in Theorem~\ref{Theorem1}.
Therefore
$$P_0=\sum_{i=1}^N e^{-{\mu_i (m+1)}}\le e^{-\mu}\sum_{i=1}^N e^{-{\mu_i m}}\le 1-\frac{\mu}{2}.$$

First, notice that $\alpha_k\le \tau$ and  $\gamma_k\le \tau$
for all $k\in [N]$. Let $B$ be a positive integer that will be determined later.
If we use $\Delta:=B+1+\lceil \log (\tau+1)\rceil$
bits for the representation of $e^{-\mu_im}$, $\alpha_k$ and $\gamma_k$, i.e.,
the absolute error is less than $2^{-B}$,
then the absolute error in computing
$r_k:=e^{-\mu_im}\gamma_k^{\ell}\alpha_k^{Z_{\ell,k}}$ is at most $O(\ell\tau^{2\ell}2^{-B})$
and in computing $P_\ell=\sum_k r_k$ is at most $O(N\ell\tau^{2\ell}2^{-B})=O(Nm\tau^{2m}2^{-B})$.
This error is at most $\mu/(4m)$ when
$B\ge 2m\log \tau+\log N+2\log m+\log(1/\mu)+2$.
Notice that, since the absolute error is less than $\mu/(4m)$,
we have $P_{\ell+1}\le P_\ell +\mu/(4m)$.
Since $P_0\le 1-\mu/2$,
we get $$P_{m+1}\le 1-\frac{\mu}{2}+\sum_{i=1}^{m+1} \frac{\mu}{4m}<1$$ which,
as shown in the proof of Theorem~\ref{Theorem1}, gives the
required bound.

Now arithmetic computations with $\Delta=B+1+\lceil \log (\tau+1)\rceil$ bit numbers take
bit-time $\tilde O(\Delta)=\tilde O(m\log \tau+\log (1/\mu))$.
Since the number of arithmetic operations
in the algorithm is $O(|S|\cdot Nm)$, the result follows.
\end{proof}

In particular we have

\begin{corollary}\label{Coro3} Let all the notation be as in Theorem~\ref{Theorem1} and  $0<\epsilon_i<1$ for all $i\in [N]$. Let $\alpha=\min_i\min(1/(1-p_i),1/(1-\lambda_i))$. For
$$m\ge \frac{3\ln N}{\min_ip_i\epsilon_i^2}$$
the algorithm runs in bit-time $T=\tilde O(|S|\cdot Nm^2\log \alpha)$ and outputs $S'=\{s_1,\ldots,s_{m+1}\}\subseteq S$ such that for
all $i\in[N']$ and $j\in(N',N]$
$$\underset{s\sim U_{S'}}{\E}[X_i(s)]\ge \lambda_i:=(1-\epsilon_i)p_i,\ \ \ \underset{s\sim U_{S'}}\E[X_j(s)]\le \lambda_j:=(1+\epsilon_j)p_j.$$
\end{corollary}
\begin{proof} Follows from the fact that if $\lambda_i=(1+\epsilon_i)p_i$ then $D(\lambda_i||p_i)\ge p_i\epsilon_i^2/2.5887\ge p_i\epsilon_i^2/3$ and if $\lambda_i=(1-\epsilon_i)p_i$ then $D(\lambda_i||p_i)\ge p_i\epsilon_i^2/2\ge p_i\epsilon_i^2/3$.

Since $\alpha_i\le 1/p_i(1-\lambda_i)$ and $\gamma_i\le 1/(1-p_i)$ we have $m\log \tau=\tilde O(m\log\alpha)$. Since $1/\mu\le \min 3/(p_i\epsilon_i^2)\le m^2$ we get $T=\tilde O(|S|\cdot Nm\cdot (m\log \tau+\log (1/\mu))=\tilde O(|S|\cdot Nm^2\log \alpha)$.
\end{proof}

\subsection{Combining with the Method of Conditional Probabilities}
In the above constructions, the time complexity is linear in $|S|$ which may be exponentially large.
In the following we get around this problem when $S$ is of the form $S_1\times \cdots\times S_n$
and the expectation of some ``intermediate'' random variables can be efficiently computed.

We prove

\begin{theorem}\label{Theo5} Let all notation and assumptions be
as in Theorem~\ref{Theorem1} and Corollary~\ref{Coro3}. Suppose $S= S_1\times S_2\times\cdots\times S_n$.
If any expectation of the form
$$\E\left[\left.{X_i(x_1,\ldots,x_n)}\ \right|\ x_1=\xi_1,\ldots,x_j=\xi_j\right]$$
can be computed in bit-time $T$ then the constructions in Theorem~\ref{Theorem1} and \ref{Theorem2}
can be performed in bit-time
$$\tilde O(T(|S_1|+\cdots+|S_n|)\cdot Nm(m\log \tau+\log (1/\mu)))$$
and in Corollary~\ref{Coro3} in bit-time
$$\tilde O(T(|S_1|+\cdots+|S_n|)\cdot Nm^2\log \alpha)$$
\end{theorem}
\begin{proof} The first result follows from the fact that since
$$\E[P'_{\ell+1}(s)]=\E_{y_1\sim S_1}[\E[P'_{\ell+1}(s)\ |\ x_1=y_1]]\le 1,$$ there is
$\xi_1\in S_1$ such that $\E[P'_{\ell+1}(s)]\ |\ x_1=\xi_1]\le 1$. So we find such $\xi_1$.
Then recursively find $\xi_2,\ldots,\xi_n$.

For this case we need the absolute error to be less than $\mu/(4mn)$ (rather than $\mu/(4m)$ as in Theorem~\ref{Theorem2}). This adds a factor of $\log n$ to the time complexity that is swallowed by the $\tilde O$.
\end{proof}

\section{Constructions of Almost $k$-Wise Independent Sets}\label{LCC}

\subsection{$\epsilon$-Balance Error-Correcting Code}
A {\it linear code} over the field $\FF_q$ is a linear
subspace $C\subset \FF_q^m$. Elements in the code
are called {\it codewords}. A linear code $C$ is called
a $[m,k,d]_q$ {\it linear code} if $C\subset \FF_q^m$
is a linear code, $|C|=q^k$ and for every two distinct codewords $v$ and $u$
in the code we have $\dist(v,u):=|\{i\ |\ v_i\not=u_i\}|\ge d$.
The latter is equivalent to: For every nonzero codeword $w$ in $C$,
we have $\wt(w):=|\{i\ |\ w_i\not=0\}|\ge d$.

A linear code $C$ is called
a $[m,k,(1-1/q)m]_q$ {\it $\epsilon$-balance error-correcting linear code} if $C\subset \FF_q^m$
is a linear code, $|C|=q^k$ and for every nonzero codeword $w\in C$ and $\xi\in \FF_q$
we have
$$(1-\epsilon)\frac{m}{q}\le |\{w_i\ |\ w_i=\xi\}|\le (1+\epsilon) \frac{m}{q}.$$
We show

\begin{lemma}\label{LBC} \label{code2} Let $q$ be a prime power, $m$ and $k$ positive integers
and  $0\le \epsilon\le 1/2$. For
$$m\ge O\left(\frac{k q \log q}{\epsilon^2}\right)$$ there is an $[m,k,(1-1/q)m]_q$
$\epsilon$-balance error-correcting linear code
that can be constructed in bit-time complexity $\tilde O(q^{k+3}/\epsilon^4)$.
\end{lemma}
\begin{proof} We use Corollary~\ref{Coro3}.
Consider $S=\FF_q^k$ with the uniform distribution.
Define for every $v\in \FF_q^{k}$ of the form
$v=(v_1,\ldots,v_j,1,0,\ldots,0)$, $j= 0,1,\ldots,k-1$, every
$\xi\in\FF_q$ and every $t\in\{1,2\}$ a random variable $X_{v,\xi,t}:\FF_q^k\to \{0,1\}$
where
$X_{v,\xi,t}(w)=[v_1w_1+\cdots+v_k w_k=\xi]$. That is, $X_{v,\xi,t}(w)=1$ if $v_1w_1+\cdots+v_k w_k=\xi$
and zero otherwise. The
number of random variables is $N=2q(q^k-1)/(q-1)$ and
$\E[X_{v,\xi,t}]=1/q$ for all $v,\xi$ and $t$.
The random variables satisfy the condition in Theorem~\ref{Theo5} with $S_i=\FF_q$ for $i=1,\ldots,k$.
Therefore an $S'=\{s_1,\ldots,s_m\}\subseteq S$ of size
$$m\ge \frac{3q\ln N}{\epsilon^2}=O\left(\frac{k q \log q}{\epsilon^2}\right)$$
that satisfies $\E_{s\sim U_{S'}}[X_{v,\xi,1}(s)]\ge (1-\epsilon)/q$
and $\E_{s\sim U_{S'}}[X_{v,\xi,2}(s)]\le (1+\epsilon)/q$ for all $v$ and $\xi$,
can be constructed in bit-time complexity
$$\tilde O\left((kq)\left(2q\frac{q^{k}-1}{q-1}\right)\frac{k^2q^2\log^2 q}{\epsilon^4}\right)=\tilde O
\left(\frac{q^{k+3}}{\epsilon^4}\right).$$
Now, $C=\{(us_1,\ldots,us_m)|u\in \FF_q^k\}$ is the code.
\end{proof}

\subsection{$\epsilon$-Bias Sample Space}\label{ebss}
Let $D$ be a probability distribution over $\FF_2^n$.
The {\it bias}
of $D$ with respect to a set of indices $I\subseteq [n]$
is defined as
$$\bias_I(D)= \left|\underset{x\sim D}{\Pr} \left(\sum_{i\in I}x_i=0\right)-\underset{x\sim D}{\Pr}\left(
\sum_{i\in I}x_i=1\right)\right|.$$
We say that $D$ is {\it $\epsilon$-bias sample space} if $\bias_I(D)\le \epsilon$ for
all non-empty subset $I\subseteq [n]$.
If $D$ is the uniform distribution over a multiset $S\subseteq \FF_2^n$ then we call
$S$ an {\it $\epsilon$-bias set}. The goal is to construct a small $\epsilon$-bias set
in polynomial time in~$n/\epsilon$. The following constructions are known from
the literature

\begin{center}
\begin{tabular}{|l|c|}
Reference & Size $|S|=O(\cdot )$\\
\hline
Alon et. al. \cite{AGHP90} & $\frac{n^2}{\epsilon^2\log^2(n/\epsilon)}$\\
\hline
AGC+Hadamard code\footnote{See the construction in \cite{BS13}}
& $\frac{n}{\epsilon^3\log(1/\epsilon)}$\\
\hline
Ben-Aroya and Ta-Shma \cite{BS13} & $\frac{n^{5/4}}{\epsilon^{2.5}\log^{5/4}(1/\epsilon)}$\\
\hline
\end{tabular}
\end{center}

The best lower bound for the size of $\epsilon$-bias set is~\cite{AGHP90,A09}
$$\Omega\left(\frac{n}{\epsilon^2\log(1/\epsilon)}\right).$$

Let $C$ be an $\epsilon$-balance error-correcting linear code
$[m,n,m/2]_2$ over $\FF_2$ with a $m\times n$ generator matrix $A$.
It is easy to see that the set of rows of $A$ is $\epsilon$-bias set
of size $m$. Therefore, by Lemma~\ref{LBC}, for $q=2$, we have

\begin{lemma}\label{EBS}  An $\epsilon$-bias set $S\subseteq \FF_2^n$ of size
$$O\left(\frac{n}{\epsilon^2}\right)$$
can be constructed in time $O(2^n/\epsilon^4)$.
\end{lemma}

\noindent {\bf Remark}: Using the powering construction in \cite{AGHP90}
with $b_{ij}=(bin(v_jx^i),bin(y))$ where $\{y\}$ is an $\epsilon$-bias set $S'\subseteq \FF_2^m$
(rather than all the elements of $\FF_2^m$) gives a polynomial time
construction of an $\epsilon$-bias set $S\subseteq \FF_2^n$ of size $O(n/\epsilon^3)$.

\subsection{$k$-wise Approximating Distributions in Time $O(n^k)$}\label{ST}

The {\it distance in the $L_p$-norm between two probability distribution}
$D$ and $Q$ over the sample space $S$ is
$$\|D-Q\|_p=\left(\sum_{s\in S}|D(s)-Q(s)|^p\right)^{\frac{1}{p}}$$
for $p\in \Re$ and $\|D-Q\|_p=\max_{s\in S}|D(s)-Q(s)|$ for $p=\infty$.

Let $S=\Sigma^n$. A {\it uniform distribution $U$ over} $S$
is a distribution where $U(s)=1/|\Sigma|^n$ for all $s \in S$.
A {\it product distribution $D$
over} $S$ is a distribution where $D(s_1,\ldots,s_n)=p_{1,s_1}\cdots p_{n,s_n}$
where $0\le p_{i,s_i}\le 1$ for all $i\in [n]$ and $s_i\in\Sigma$.

For $I=(i_1,\ldots,i_k)$ where
$1\le i_1<\cdots<i_k\le n$ the {\it distribution $D_I$ restricted to} $I$
over $\Sigma^k$ is $D_I(\sigma_1,\ldots,\sigma_k)=\Pr_{s\sim D}[s_{i_1}=\sigma_1\wedge
\cdots\wedge s_{i_k}=\sigma_k]$. Two distributions $D$ and $Q$ over $\Sigma^n$ are called
{\it $k$-wise $\epsilon$-close in the $L_p$-norm}, if for any
$I=(i_1,\ldots,i_k)$, $\|D_I-Q_I\|_p\le \epsilon$. If $D$ is the uniform distribution then
$Q$ is called {\it $\epsilon$-almost $k$-wise independent in the $L_p$-norm}.

When $\epsilon=0$ then $S$ is called $k$-{\it wise independent set}. It is known that the size of any $k$-wise independent set is $n^{\Theta(k)}$, \cite{ABI,CGH}. See also \cite{AGM03}.

The goal is: given a distribution $D$. Construct $S'\subset S$ of small size such that the
uniform distribution on $S'$ is $k$-wise $\epsilon$-close to $D$ in the $L_p$-norm.
We will just say that $S'$ is $k${\it -wise $\epsilon$-close} to $D$ in the $L_p$-norm and if $D$ is the uniform distribution we say  that $S$ is {\it $\epsilon$-almost $k$-wise independent in the $L_p$-norm}.

For $\Sigma=\{0,1\}$, Naor and Naor proved
\begin{lemma}\label{NN} {\rm \cite{NN93}}. Let $k<n$ be an odd integer, $t$
is a power of $2$ and $$n\le 2^{\lfloor 2(m-1)/(k-1)\rfloor}-1.$$
Given an $\epsilon$-bias set $S\subseteq \{0,1\}^m$ of size $t$, one can,
in polynomial time, construct
a set $R\subseteq \{0,1\}^n$ of size $t$ that is $\epsilon$-almost $k$-wise independent in the
$L_\infty$-norm and $2^{k/2}\epsilon$-almost $k$-wise independent in
the $L_1$-norm.
\end{lemma}

For $\Sigma=\{0,1\}$, the following are the best known polynomial time constructions
of sets that are $\epsilon$-almost $k$-wise independent in the $L_\infty$-norm and $L_1$-norm.
The constructions use the $\epsilon$-bias sets in Section~\ref{ebss}
with Lemma~\ref{NN}. See also the sizes without the small terms $k$, $\log(1/\epsilon)$ and $\log\log n$ in the first three rows of the table in the introduction.

\begin{center}
\begin{tabular}{|c|c|c|}
Reference & Size for $L_\infty$ & Size for $L_1$\\
\hline
\cite{AGHP90} & $\frac{k^2\log^2n}{\epsilon^2(\log^2k+(\log\log n)^2+\log^2(1/\epsilon)))}$
& $\frac{k^2 2^k\log^2n}{\epsilon^2(k^2+(\log\log n)^2+\log^2(1/\epsilon))}$\\
\hline
AGC+Hadamard& $\frac{k\log n}{\epsilon^3\log(1/\epsilon)}$
&$\frac{k 2^{3k/2}\log n}{\epsilon^3(k+\log(1/\epsilon))}$\\
\hline
\cite{BS13} & $\frac{k^{5/4}\log^{5/4} n}{\epsilon^{2.5}\log^{5/4}(1/\epsilon)}$
&$\frac{k^{5/4}2^{5k/4}\log^{5/4} n}{\epsilon^{2.5}(k^{5/4}+\log^{5/4}(1/\epsilon))}$\\
\hline
\end{tabular}
\end{center}

By Lemma \ref{EBS} and Lemma~\ref{NN} we have
\begin{lemma}\label{Li5} An $\epsilon$-almost $k$-wise independent set in the
$L_\infty$-norm of size
$$O\left(\frac{k\log n}{\epsilon^2}\right)$$ can be constructed in time $O(n^{(k-1)/2}/\epsilon^4)$.

An $\epsilon$-almost $k$-wise independent in
the $L_1$-norm of size
$$O\left(\frac{k2^k\log n}{\epsilon^2}\right)$$
can be constructed in time $O(2^{2k}n^{(k-1)/2}/\epsilon^4)$.
\end{lemma}

We now prove
\begin{theorem}\label{L6} Let $\epsilon<1/2^k$. An $\epsilon$-almost $k$-wise independent set in the
$L_\infty$-norm of size
$$O\left(\frac{k\log n}{2^k\epsilon^2}\right)$$ can be constructed in time $\tilde O(n^k/\epsilon^4)$.
\end{theorem}
\begin{proof} Consider $S=\{0,1\}^n$  with the uniform distribution.
Define for every $1\le i_1<i_2<\cdots<i_k\le n$
and $(\xi_1,\ldots,\xi_k)\in\{0,1\}^k$ the random variable
$X_{i_1,\ldots,i_k,\xi_1,\ldots,\xi_k}(s)$ that is equal to $1$ if and only if
$s_{i_j}=\xi_j$ for all $j=1,\ldots,k$. Now the result follows from Corollary~\ref{Coro3} and Theorem~\ref{Theo5}.
\end{proof}
This gives an $\epsilon$-almost $k$-wise independent set in the $L_1$-norm
of size $O(k2^k\log n/\epsilon^2)$.
We now give a better bound. We first prove

\begin{lemma} \label{L07} Let $0\le r<k$. An $\epsilon$-almost $k$-wise independent set in the
$L_1$-norm of size
$$m=O\left(\frac{2^k+k 2^r\log n}{\epsilon^2}\right)$$ can be constructed in time $\tilde O(2^{2^{k-r}}n^k/\epsilon^4)$.
\end{lemma}
\begin{proof} Consider $S=\{0,1\}^n$  with the uniform distribution.
For every $a\in \{0,1\}^r$ and $B\subseteq \{0,1\}^{k-r}$
we define the random variable
$Z_{i_1,\ldots,i_k,a,B}(s)=1$ if and only if $(s_{i_1},\ldots,s_{i_k})\in \{a\}\times B$.
Let $S'\subset S$ and suppose for every $I=(i_1,\ldots,i_k)$, $a\in \{0,1\}^r$ and $B\subseteq \{0,1\}^{k-r}$
we have $$|\E_{s\sim U_{S'}}[Z_{I,a,B}]-\E_{s\sim U_{S}}[Z_{I,a,B}]|=\left|\E_{s\sim U_{S'}}[Z_{I,a,B}]-\frac{|B|}{2^k}\right|\le \frac{\epsilon}{2^{r+1}}.$$
Then (here $\E$ is $\E_{s\sim U_{S'}}$)
\begin{eqnarray*}
\sum_{a\in \{0,1\}^r,b\in \{0,1\}^{k-r}}\left|\E[Z_{I,a,\{b\}}]-\frac{1}{2^k}\right|
&=&\sum_{a\in \{0,1\}^r}\sum_{b\in \{0,1\}^{k-r}}\left|\E[Z_{I,a,\{b\}}]-\frac{1}{2^k}\right|\\
&=&\sum_{a\in \{0,1\}^r}\max_{B\subseteq \{0,1\}^{k-r}} \left(\E[Z_{I,a,B}]-\frac{|B|}{2^k}\right)+\left(\frac{|\bar B|}{2^k}-\E[Z_{I,a,\bar B}]\right)\\
&\le& \sum_{a\in \{0,1\}^r}\frac{\epsilon}{2^r}=\epsilon,
\end{eqnarray*}
and therefore, $S'$ is an $\epsilon$-almost $k$-wise
independent in the $L_1$ norm.

Now to construct such a set $S'$ we use Corollary~\ref{Coro3} and Theorem~\ref{Theo5}.
We have $N=2^{2^{k-r}+1}2^r{n\choose k}$ and define for each $Z_{I,a,B}$,
$\epsilon_{I,a,B}=\epsilon 2^{k-r-1}/|B|$ and $p_{I,a,B}=|B|/2^k$.
By Theorem~\ref{Theorem1},
$$m\ge \frac{\ln 2^{2^{k-r}+1}2^r{n\choose k}}{\min p_{I,a,B}\epsilon_{I,a,B}^2}=O\left(\frac{2^k+k2^r\log n}{\epsilon^2}\right).$$
\end{proof}

We now prove
\begin{theorem} \label{L07E} Let $d>0$ be any real number.
An $\epsilon$-almost $k$-wise independent set in the
$L_1$-norm of size
$$m=O\left(\frac{2^k+(k/d)2^k+k \log n}{\epsilon^2}\right)$$ can be constructed in time $\tilde O(n^{k+d}/\epsilon^4)$.

In particular, an $\epsilon$-almost $k$-wise independent set in the
$L_1$-norm of size
$$m=O\left(\frac{2^k+k \log n}{\epsilon^2}\right)$$ can be constructed in time $\tilde O(n^{2k}/\epsilon^4)$.
\end{theorem}
\begin{proof} Follows from Lemma~\ref{L07} with $r=\max(\lceil k-\log\log n-\log d\rceil,0)$.
\end{proof}

%

The above results can be extended to any product distribution over any alphabet.

\subsection{Efficient Construction for Any $k$}\label{Theorems}
In this subsection, we give a construction that is efficient for
any $k$. We will give the results for the uniform distribution.
Similar results can be obtained for the product distribution.

We first define the dense perfect hash family.
We say that $H\subseteq [q]^n$ is a $(1-\epsilon)$-{\it dense} $(n,q,k)$-{\it perfect hash family}
if for every
$1\le i_1<i_2<\cdots<i_k\le n$ there are at least
$(1-\epsilon)|H|$ elements $h\in H$ such that $h_{i_1},\ldots,h_{i_k}$ are distinct.

We will use the following lemma. The proof is in~\cite{B14b} for a power of prime $q$. See also \cite{PR11}. Here we give the proof for any $q$.

\begin{lemma}~\label{Den} If $n>q>4k^2/\epsilon$ then
there is a $(1-\epsilon)$-dense $(n,q,k)$-PHF
of size
$$O\left(\frac{k^2\log n}{\epsilon\log(\epsilon q/k^2)}\right)$$
that can be constructed in polynomial time.
\end{lemma}
\begin{proof} We use Theorem~\ref{Theorem1}. The sample space is $S=[q]^n$
with the uniform distribution. The random variables are $X_{i,j}(s)=I[s_i=s_j]$ for all $1\le i<j\le n$.
That is, $X_{i,j}(s)=1$ if $s_i=s_j$ and $X_{i,j}(s)=0$ otherwise. We have
$p=\E[X_{i,j}]=1/q$. The number of such random variables is $N={n\choose 2}$.
Let $h=\lceil k^2/\epsilon\rceil$ and $\lambda=1/h>p$. Then
\begin{eqnarray*}
D(\lambda||p)&=&\frac{1}{h}\ln\left(\frac{1/h}{1/q}\right)+\left(1-\frac{1}{h}\right)\ln\frac{1-1/h}{1-1/q}\\
&\ge&\frac{1}{h}\ln q-\frac{1}{h}\ln h+\left(1-\frac{1}{h}\right)\ln\left(1-\frac{1}{h}\right)\\
&\ge&\frac{\ln q-\ln h-1}{h}.
\end{eqnarray*}
By Theorem~\ref{Theorem1}, in polynomial time, we can find a multiset $S'=\{s_1,\ldots,s_m\}\subset S$ where
$$m=\frac{\ln N}{D(\lambda||p)}=O\left(\frac{k^2\log n}{\epsilon\log(\epsilon q/k^2)}\right)$$
such that for all $1\le i\le j\le n$ we have $\Pr_{s\in U_{S'}}[s_{i}=s_{j}]\le 1/h$.
Then for any $1\le i_1<i_2<\cdots<i_k\le n$ we have
\begin{eqnarray*}
\Pr_{s\in U_{S'}}[ (\exists 1\le j_1<j_2\le k) \ s_{i_{j_1}}=s_{i_{j_2}} ]\le \frac{{k\choose 2}}{h}\le \epsilon.
\end{eqnarray*}

\end{proof}

We are now ready to give our main three results.

We first show

\begin{theorem} Let $\epsilon<1/2^k$. An $\epsilon$-almost $k$-wise independent set in the $L_\infty$-norm
of size
$$O\left(\frac{k^3\log n}{2^k \epsilon^3}\right)$$
can be constructed in polynomial time.
\end{theorem}
\begin{proof} In this proof, $c_i$, $i=1,2,3,\cdots$, are some constants.
For $n\le (k/\epsilon)^{3k}$, the size of the $\epsilon$-almost $k$-wise independent set
in~\cite{AGHP90} (see the table) is at most
$$\frac{c_1k^2\log^2n}{\epsilon^2\log^2(1/\epsilon)}\le \frac{c_2k^3\log n}{\epsilon^2\log(1/\epsilon)}
\le \frac{c_2k^3\log n}{2^k\epsilon^3}.$$

Now let $n\ge (k/\epsilon)^{3k}$. Let $q=\lceil (k/\epsilon)^{3}\rceil$ and $H\subseteq [q]^n$ be a $(1-\epsilon/4)$-dense $(n,q,k)$-PHF of size
$$|H|=\frac{c_3k^2\log n}{\epsilon \log(\epsilon q/k^2)}\le \frac{c_4k^2\log n}{\epsilon \log(1/\epsilon)}.$$
Such an $H$ exists by Lemma~\ref{Den} and can be constructed in polynomial time.
By Theorem~\ref{L6}, an $(\epsilon/4)$-almost $k$-wise independent set
$R\subseteq \{0,1\}^q$ in the $L_\infty$-norm
of size
$$|R|=\frac{c_5 k\log q}{2^k \epsilon^2}\le \frac{c_6 k\log (1/\epsilon)}{2^k \epsilon^2}$$
can be constructed in time $\tilde O(q^k/\epsilon^4)=poly(n,1/\epsilon)$.
Define $$S= \{(v_{u_1},\ldots,v_{u_n})\ |\ u\in H, v\in R\}.$$

The size of $S$ is
$$|R|\cdot|H|=\frac{c_7 k^3\log n}{2^k\epsilon^3}.$$
Now for a random uniform $s\in S$, any $1\le i_1<\cdots<i_k\le n$ and $\xi\in\{0,1\}^k$,
\begin{eqnarray*}
\Pr_{s\in S} [(s_{i_1},\ldots,s_{i_k})=\xi]&=&\Pr_{v\in R,u\in H} [(v_{u_{i_1}},\ldots,v_{u_{i_k}})=\xi]\\
&\le& \Pr_{v\in R,u\in H} [(v_{u_{i_1}},\ldots,v_{u_{i_k}})=\xi\ |\ u_{i_1},\ldots,u_{i_k}\mbox{\ are distinct}]+\\
&& \ \ \  \Pr_{v\in R,u\in H}[u_{i_1},\ldots,u_{i_k}\mbox{\ are not distinct}]\\
&\le& \left(\frac{1}{2^k}+\frac{\epsilon}{4}\right)+\frac{\epsilon}{4}\le \frac{1}{2^k}+\epsilon
\end{eqnarray*}
and
\begin{eqnarray*}
\Pr_{s\in S} [(s_{i_1},\ldots,s_{i_k})=\xi]&=&\Pr_{v\in R,u\in H} [(v_{u_{i_1}},\ldots,v_{u_{i_k}})=\xi]\\
&\ge& \Pr_{v\in R,u\in H} [(v_{u_{i_1}},\ldots,v_{u_{i_k}})=\xi\ |\ u_{i_1},\ldots,u_{i_k}\mbox{\ are distinct}]\cdot\\
&& \ \ \  \Pr_{v\in R,u\in H}[u_{i_1},\ldots,u_{i_k}\mbox{\ are distinct}]\\
&\ge& \left(\frac{1}{2^k}-\frac{\epsilon}{4}\right)\left(1-\frac{\epsilon}{4}\right)\ge \frac{1}{2^k}-\epsilon.
\end{eqnarray*}
Therefore, $S$ is $\epsilon$-almost $k$-wise independent set in $L_\infty$-norm.
\end{proof}

For the $L_1$-norm we prove
\begin{theorem} An $\epsilon$-almost $k$-wise independent set in the $L_1$-norm
of size
$$\tilde O\left(\frac{2^k+\log n}{ \epsilon^3}\right)$$
can be constructed in polynomial time.
\end{theorem}

This theorem follows from the following two results

\ignore{
\begin{theorem} An $\epsilon$-almost $k$-wise set in the $L_1$-norm
of size
$$O\left(\frac{k^3 2^k\log n}{ \epsilon^3}\right)$$
can be constructed in polynomial time.
\end{theorem}
\begin{proof} In this proof, $c_i$, $i=1,2,3,\cdots$, are some constants.
For $n\le (k/\epsilon)^{3k}$, the size of the $\epsilon$-almost $k$-wise set
in~\cite{AGHP90} (see the table) is at most
$$\frac{c_1k^22^k\log^2n}{\epsilon^2 (k^2+\log^2(1/\epsilon))}\le
\frac{c_2k(\log k)2^k\log n}{\epsilon^2\log(1/\epsilon)}
\le \frac{c_2k^3 2^k\log n}{\epsilon^3}.$$

Now let $n\ge (k/\epsilon)^{3k}$. Let $(k/\epsilon)^{3}\le q\le 2(k/\epsilon)^{3}$.
Let $H\subseteq [q]^n$ be a $(1-\epsilon/4)$-dense $(n,q,k)$-PHF of size
$$|H|=\frac{c_3k^2\log n}{\epsilon \log(\epsilon q/k^2)}\le \frac{c_4k^2\log n}{\epsilon \log(k/\epsilon)}.$$
By Lemma~\ref{Li5}, an $(\epsilon/4)$-almost $k$-wise independent set
$R\subseteq \{0,1\}^q$ in the $L_\infty$-norm
of size
$$|R|=\frac{c_5 2^k k\log q}{ \epsilon^2}\le \frac{c_6 k 2^k\log (k/\epsilon)}{ \epsilon^2}$$
can be constructed in time $2^{2k}q^{(k-1)/2}/\epsilon^4=poly(n,1/\epsilon)$.
Define $$S= \{(v_{u_1},\ldots,v_{u_n})\ |\ u\in H, v\in R\}.$$

The size of $S$ is
$$|R|\cdot|H|=\frac{c_7 k^3 2^k\log n}{\epsilon^3}.$$
It is easy to show that $S$ is $\epsilon$-almost $k$-wise independent set in $L_1$-norm.
\end{proof}}

\begin{lemma} For $n>2^{2^k}$, an $\epsilon$-almost $k$-wise independent set in the $L_1$-norm
of size
$$O\left(\frac{k^3 \log^{1/2}(1/\epsilon)\log n}{ \epsilon^3}\right)=\tilde O\left(\frac{\log n}{\epsilon^3}\right)$$
can be constructed in polynomial time.
\end{lemma}
\begin{proof}
Let $n\ge 2^{2^k}$. If $\epsilon<1/2^{2^k/k}$ then the AGC+Hadamard construction is
of size (see the table)
$$\frac{c_1 k2^{3k/2}\log n}{\epsilon^3(k+\log(1/\epsilon))}\le \frac{c_2 k^{2.5}\log^{1/2}(1/\epsilon)\log n}{\epsilon^3}.$$

Now let $\epsilon\ge 1/2^{2^k/k}$.
Let $q=\lceil 2^{2^k/k}k^2/\epsilon\rceil$ and $H\subseteq [q]^n$ be a $(1-\epsilon/4)$-dense $(n,q,k)$-PHF of size
$$\frac{c_3k^2\log n}{\epsilon \log(\epsilon q/k^2)}\le \frac{c_4k^3\log n}{\epsilon 2^k}.$$
Such an $H$ exists by Lemma~\ref{Den} and can be constructed in polynomial time.
By Theorem~\ref{L07E}, an $(\epsilon/4)$-almost $k$-wise independent set
$R\subseteq \{0,1\}^q$ in the $L_1$-norm
of size
$$\frac{c_5 (2^k+k\log q)}{\epsilon^2}\le \frac{c_6 2^k}{\epsilon^2}$$
can be constructed in time $\tilde O(q^{2k}/\epsilon^4)=poly(n,1/\epsilon)$.
Define $$S= \{(v_{u_1},\ldots,v_{u_n})\ |\ u\in H, v\in R\}.$$

The size of $S$ is
$$|R|\cdot|H|=\frac{c_7 k^3\log n}{\epsilon^3}.$$
Now for a random uniform $s\in S$, any $1\le i_1<\cdots<i_k\le n$ and any boolean function $f:\{0,1\}^k\to \{0,1\}$,
\begin{eqnarray*}
\Pr_{s\in S} [f(s_{i_1},\ldots,s_{i_k})=1]&=&\Pr_{v\in R,u\in H} [f(v_{u_{i_1}},\ldots,v_{u_{i_k}})=1]\\
&\le& \Pr_{v\in R,u\in H} [f(v_{u_{i_1}},\ldots,v_{u_{i_k}})=1\ |\ u_{i_1},\ldots,u_{i_k}\mbox{\ are distinct}]+\\
&& \ \ \  \Pr_{v\in R,u\in H}[u_{i_1},\ldots,u_{i_k}\mbox{\ are not distinct}]\\
&\le& \left(\Pr[f=1]+\frac{\epsilon}{4}\right)+\frac{\epsilon}{4}\le \Pr[f=1]+\epsilon.
\end{eqnarray*}
Therefore $S$ is $\epsilon$-almost $k$-wise independent set in $L_\infty$-norm.
\end{proof}

\begin{lemma} For $n\le 2^{2^k}$, an $\epsilon$-almost $k$-wise independent set in the $L_1$-norm
of size
$$\tilde O\left(\frac{2^k}{ \epsilon^3}\right)$$
can be constructed in polynomial time.
\end{lemma}
\begin{proof}
If $n<(4k^2/\epsilon)^{2k}$ then we use the bound in \cite{AGHP90} (see the table) and get an $\epsilon$-almost $k$-wise independent set in the $L_1$-norm
of size
$$O\left(\frac{k^2\min(k^2,\log^2(1/\epsilon))\cdot  2^k}{\epsilon^2}\right)=\tilde O\left(\frac{2^k}{\epsilon^2}\right).$$
Otherwise, we take $q=\lceil n^{1/k}\rceil\ge (4k^2/\epsilon)^2$ and use Theorem~\ref{L07E}
to construct an $(\epsilon/4)$-almost $k$-wise set in the $L_1$-norm, $R\subseteq \{0,1\}^q$,
of size (notice that $\log n\le 2^k$) $$O\left(\frac{2^k}{\epsilon^2}\right)$$ in time $\tilde O(q^{2k}/\epsilon^4)=poly(n,1/\epsilon)$.
Composing this with the $(\epsilon/4)$-dense $(n,q,k)$-PFH in Lemma~\ref{Den} of size
$$O\left(\frac{k^2\log n}{\epsilon \log(\epsilon q/k^2)}\right)=O\left(\frac{k^3}{\epsilon}\right) $$
gives the result.
\end{proof}

\section{Lower Bounds}\label{LowerB}
In this section we give the following two lower bounds

\begin{theorem}\label{LB1} Let $1/poly(n)\le \epsilon<1 /2^{k+1}$ and $k<n/2$.
Any $\epsilon$-almost $k$-wise independent set in the $L_\infty$-norm is of size
$$\Omega\left(\frac{\log n}{2^k\epsilon^2 \log(1/2^{k+1}\epsilon)}\right)
=\tilde \Omega\left(\frac{\log n}{2^k\epsilon^2 }\right).$$
\end{theorem}

\begin{theorem}\label{LB2}
Let $\epsilon>1/n^{k/5}$. Any $\epsilon$-almost $k$-wise independent set in the $L_1$-norm is of size
$$\Omega\left(\frac{k\log n}{\epsilon^2 \log(1/\epsilon)}\right)
=\tilde \Omega\left(\frac{\log n}{\epsilon^2 }\right).$$
\end{theorem}

The following is proved in \cite{AA07,A09}. We give the proof for completeness.
\begin{lemma}~\label{MRRW} Let $S\subset \{0,1\}^n$ and $r<n$ be an even number.
If for every distinct $i_1,\ldots,i_r \in [n]$, $\alpha\in \{0,1\}^r\backslash\{0^r\}$
and a random uniform $s\in S$ we have $1/2+\epsilon\ge \Pr[\alpha_1 s_{i_1}+\cdots+\alpha_r s_{i_r}]\ge 1/2-\epsilon$
then $$|S|\ge
\Omega\left(\min\left(\frac{r\log (n/r)}{\epsilon^2 \log (1/\epsilon)},2^{r/2}{n\choose r/2}\right)\right).$$

In particular, for $\epsilon>1\left/\left(2^{r/4}{n\choose r/2}^{1/2}\right)\right.$
$$|S|\ge
\Omega\left(\frac{r\log (n/r)}{\epsilon^2 \log (1/\epsilon)}\right).$$
\end{lemma}
\begin{proof} Let $S=\{s^{(1)},\ldots,s^{(m)}\}$ and $I=\{(i_1,\ldots,i_{r/2})\ |\ 1\le i_1<i_2<\cdots<i_{r/2}\le 1\}$.
Consider
$$C=\{(\beta_1s^{(1)}_{i_1}+\cdots+\beta_{r/2}s^{(1)}_{i_{r/2}},\ldots,
\beta_1s^{(m)}_{i_1}+\cdots+\beta_{r/2}s^{(m)}_{i_{r/2}})\ |\ i\in I,\beta\in\{0,1\}^{r/2}\backslash\{0^{r/2}\}\}.$$
Then for two distinct $u,v\in C$ we have (for some $i_1,\ldots,i_r$ and $(\alpha_1,\ldots,\alpha_r)\in \{0,1\}^r\backslash\{0^r\}$)
$$dist(u,v)=wt(u+v)=m\cdot \Pr_{s\in U_S}[\alpha_1 s_{i_1}+\cdots+\alpha_r s_{i_r}]\in
\left[ \left(\frac{1}{2}-\epsilon\right)m,\left(\frac{1}{2}+\epsilon\right)m\right].$$
Therefore $C$ is an $\epsilon$-balance error-correcting code size $(2^{r/2}-1){n\choose r/2}$. By MRRW bound,\cite{MRRW}, for binary code with the results in Section 7 and (3)
in~\cite{AGHP90} and the bound in~\cite{A09}, the result follows.
\end{proof}

We now prove Theorem~\ref{LB1}
\begin{proof} Let $S$ be a $\epsilon$-almost $k$-wise set in the $L_\infty$ norm.
Let $r\ge 2$ be an even constant such that $\epsilon>1\left/\left(2^{r/4}{n-k+r\choose r/2}^{1/2}\right)\right.$.
For $\xi\in\{0,1\}^{k-r}$
consider the sets
$S_\xi=\{s\in S\ |\ (s_1,\ldots,s_{k-r})=\xi\}$. Obviously, $\{S_\xi\}_{\xi\in\{0,1\}^{k-r}}$
is a partition of $S$.

Let $I=\{k-r+1,k-r+2,\ldots,n\}$.
For distinct $i_1,\ldots,i_r\in I$, $\alpha\in\{0,1\}^r\backslash\{0^r\}$, $\xi=\xi_1,\ldots,\xi_{k-r}\in\{0,1\}^{k-r}$,
and a random uniform $x\in S$ we have
\begin{eqnarray*}
\Pr[\alpha_1x_{i_1}+\cdots+\alpha_rx_{i_r}=1 |\ x\in S_\xi]&=&\frac{\Pr[\alpha_1x_{i_1}+\cdots+\alpha_rx_{i_r}=1,(x_1,\ldots,x_{k-r})=\xi]}
{\Pr[(x_1,\ldots,x_{k-r})=\xi]}\\
&=& \frac{\Pr[\alpha_1x_{i_1}+\cdots+\alpha_rx_{i_r}=1,(x_1,\ldots,x_{k-r})=\xi]}
{\sum_{u\in\{0,1\}^r}\Pr[(x_{i_1},\ldots,x_{i_r})=u,(x_1,\ldots,x_{k-r})=\xi]}\\
&\ge& \frac{2^{r-1}\left(\frac{1}{2^k}-\epsilon\right)}{2^r\left(\frac{1}{2^k}+\epsilon\right)}\ge \frac{1}{2}-2^{k+1}\epsilon.
\end{eqnarray*}
In the same way $$\Pr[\alpha_1x_{i_1}+\cdots+\alpha_rx_{i_r}=1 |\ x\in S_\xi]\le \frac{1}{2}+2^{k+1}\epsilon.$$
Therefore, by Lemma~\ref{MRRW}, for $\epsilon>1\left/\left(2^{r/4}{n-k+r\choose r/2}^{1/2}\right)\right.$,
$$|S|=\sum_{\xi\in \{0,1\}^{k-r}}|S_\xi|= 2^{k-r}\cdot \Omega\left(\frac{r\log(|I|/r)}{(2^{k+1}\epsilon)^2\log(1/(2^{k+1}\epsilon))}\right)
=\Omega\left(\frac{\log n}{2^k\epsilon^2 \log(1/2^{k+1}\epsilon)}\right).$$
\end{proof}

We now prove Theorem~\ref{LB2}
\begin{proof}
For every distinct $i_1,\ldots,i_k\in [n]$ and $\alpha\in\{0,1\}^k\backslash\{0^k\}$
and any random uniform $x\in S$
\begin{eqnarray*}
\left|\Pr[\alpha_1x_{i_1}+\cdots+\alpha_k x_{i_k}=1]-\frac{1}{2}\right|&\le& \sum_{\alpha_1\xi_1+\cdots+\alpha_k\xi_k=1} \left|\Pr[x_{i_1}=\xi_1,\ldots,x_{i_k}=\xi_k]-\frac{1}{2^k}\right|\\
&\le& \sum_{\xi\in\{0,1\}^{k}} \left|\Pr[x_{i_1}=\xi_1,\ldots,x_{i_k}=\xi_k]-\frac{1}{2^k}\right|\le \epsilon.\\
\end{eqnarray*}
Therefore, by Lemma~\ref{MRRW}, for $\epsilon>1/n^{k/5}$, we have
$$|S|=\Omega\left(\frac{k\log n}{\epsilon^2\log(1/\epsilon)}\right).$$
\end{proof}

{\bf Acknowledgement.} We would like to thank Noga Alon for pointing out the paper~\cite{A09}.


\begin{thebibliography}{}


\bibitem{A09}
N. Alon.
Perturbed Identity Matrices Have High Rank: Proof and Applications.
{\it Combinatorics, Probability \& Computing}. 18(1-2). pp. 3-15. (2009).

\bibitem{AA07}
N. Alon, A. Andoni, T. Kaufman, K. Matulef, R. Rubinfeld, N. Xie.
Testing $k$-wise and almost $k$-wise independence. STOC 2007. pp. 496--505. (2007).

\bibitem{ABI}
N. Alon, L. Babai, A. Itai.
A Fast and Simple Randomized Parallel Algorithm for the Maximal Independent Set Problem. J. Algorithms 7(4). pp. 567--583. (1986).

\bibitem{ABNNR}
N. Alon, J. Bruck, J. Naor, M. Naor, R. M. Roth.
Construction of asymptotically good low-rate error-correcting codes through pseudo-random graphs.
IEEE Transactions on Information Theory 38(2). pp. 509-516. (1992)

\bibitem{AG09}
N. Alon, S. Gutner. Balanced hashing, color coding and approximate counting. Electronic Colloquium
on Computational Complexity (ECCC) 16: 12 (2009).

\bibitem{AG10}
N. Alon, S. Gutner.
Balanced families of perfect hash functions and their applications.
ACM Transactions on Algorithms. 6(3). (2010).

\bibitem{AGHP90}
N. Alon, O. Goldreich, J. Hastad, R. Peralta.
Simple constructions of almost $k$-wise independent random variables. FOCS 1990.
pp. 544--553. (1990).

\bibitem{AGM03}
N. Alon, O. Goldreich, Y. Mansour.
Almost $k$-wise independence versus $k$-wise independence.
{\it Inf. Process. Lett.} 88(3): 107-110 (2003)

\bibitem{AHK12}
S. Arora, E. Hazan, S. Kale.
The multiplicative weights update method: A meta-algorithm and applications.
Theory of computing, 8, pp. 121--164. (2012).

\bibitem{AM06}
N. Alon, D. Moshkovitz, S. Safra. Algorithmic construction of sets for $k$-restrictions. ACM
Transactions on Algorithms 2(2). pp. 153--177. (2006).

\bibitem{B14b}
N. H. Bshouty.
Linear time Constructions of some $d$-Restriction Problems.
CoRR abs/1406.2108 (2014). CIAC 2015. pp. 74--88. (2015).

\bibitem{BG}
N. H. Bshouty, A. Gabizon.
Almost Optimal Cover-Free Families. CoRR abs1507.07368. (2015).

\bibitem{BS13}
A. Ben-Aroya, A. Ta-Shma.
Constructing Small-Bias Sets from Algebraic-Geometric Codes.
{\it Theory of Computing}. 9. pp. 253--272. (2013).

\bibitem{CGH}
B. Chor, O. Goldreich, J. H\aa{a}stad, J. Friedman, S. Rudich, R. Smolensky.
The Bit Extraction Problem of t-Resilient Functions. FOCS 1985: pp. 396--407. (1985).

\bibitem{IK}
R. Impagliazzo, V. Kabanets.
Constructive Proofs of Concentration Bounds
ECCC TR10-072. (2010).

\bibitem{KM94}
D. Koller, N. Megiddo. Constructing Small Sample Spaces Satisfying Given Constraints. {\it SIAM
J. Discrete Math.} 7(2). pp. 260--274. (1994).

\bibitem{PR11}
E. Porat, A. Rothschild.
Explicit Nonadaptive Combinatorial Group Testing Schemes.
IEEE Transactions on Information Theory 57(12), pp. 7982--7989 (2011).

\bibitem{NN93}
J. Naor, M. Naor.
Small-Bias Probability Spaces: Efficient Constructions and Applications.
{\it SIAM J. Comput.} 22(4). pp. 838-856 (1993)

\bibitem{MRRW}
R. J. McEliece, E. R. Rodemich, H. Rumsey Jr., L. R. Welch.
New upper bounds on the rate of a code via the Delsarte-MacWilliams inequalities. IEEE Transactions on Information Theory 23(2): 157-166 (1977)

\bibitem{PR}
E. Porat, A. Rothschild. Explicit Nonadaptive Combinatorial Group Testing Schemes. {\it IEEE
Transactions on Information Theory}. 57(12), pp. 7982--7989. (2011).

\bibitem{NSS}
M. Naor, L. J. Schulman, A. Srinivasan.
Splitters and Near-Optimal Derandomization. FOCS 1995: 182--191. (1995).

\bibitem{WX08}
A. Wigderson, D. Xiao.
Derandomizing the Ahlswede-Winter matrix-valued Chernoff bound using pessimistic estimators, and applications.
Theory of Computing 4(1): pp. 53--76. (2008).

\end{thebibliography}
\end{document}